\documentclass[11pt]{article}

\usepackage{amssymb,amsmath,amsthm}
\usepackage{graphicx}
\usepackage{latexsym}
\usepackage{pstricks}

\newtheorem{definition}{Definition}
\newtheorem{theorem}[definition]{Theorem}
\newtheorem{corollary}[definition]{Corollary}
\newtheorem{lemma}[definition]{Lemma}

\newcommand{\XF}{{\cal F}}

\newcount\hour \newcount\minutes \hour=\time \divide\hour by 60
\minutes=\hour \multiply\minutes by -60 \advance\minutes by \time
\def\mmmddyyyy{\ifcase\month\or Jan\or Feb\or Mar\or Apr\or May\or
Jun\or Jul\or Aug\or Sep\or Oct\or Nov\or Dec\fi \space\number\day,
\number\year}
\def\hhmm{\ifnum\hour<10 0\fi\number\hour :%
  \ifnum\minutes<10 0\fi\number\minutes} 
\makeatother

\begin{document}

\title{A note on the triangle inequality for the Jaccard distance}

\author{{\em Sven Kosub}\\
Department of Computer \& Information Science, University of Konstanz\\
Box 67, D-78457 Konstanz, Germany\\
{\tt Sven.Kosub@uni-konstanz.de}}

\date{\today}

\maketitle

\begin{abstract}
Two simple proofs of the triangle inequality for the Jaccard distance in terms of nonnegative, monotone, submodular functions are given and discussed.
\end{abstract}

\noindent
The Jaccard index \cite{jaccard-1901} is a classical similarity measure on sets with a lot of practical applications in 
information retrieval, data mining, machine learning, and many more (cf., e.g., \cite{gower-2008}).
Measuring the relative size of the overlap of two finite sets $A$ and $B$, the Jaccard index $J$ and the associated Jaccard distance
$J_\delta$ are formally defined as:
\[
J(A,B)=_{\rm def} \frac{|A\cap B|}{|A\cup B|}, \qquad J_\delta(A,B)=_{\rm def} 1- J(A,B)=1-\frac{|A\cap B|}{|A\cup B|} = \frac{|A\triangle B|}{|A\cup B|}
\]
where $J(\emptyset,\emptyset)=_{\rm def} 1$.
The Jaccard distance $J_\delta$ is known to fulfill all properties of a metric, most notably, the triangle inequality---a fact 
that has been observed many times, e.g., via metric transforms
\cite{marczewski-steinhaus-1958,simovici-djeraba-2008,gardner-etal-2014}, embeddings in vector spaces (e.g., \cite{tanimoto-1958,lipkus-1999,gardner-etal-2014}), min-wise independent permutations \cite{charikar-2002}, or sometimes cumbersome arithmetics \cite{levandowsky-winter-1971,fujita-2013}. 
A very simple, elementary proof of the triangle inequality was given in \cite{gilbert-1972} using an appropriate 
partitioning of sets. 

Here, we give two more simple, direct proofs of the triangle inequality. 
One proof comes without any set difference or disjointness of sets. It is based only on the fundamental equation 
$|A\cup B|+|A\cap B|=|A|+|B|$. As such, the proof is generic and leads to (sub)modular versions of the Jaccard distance 
(as defined below). 
The second proof unfolds a subtle difference between the two possible versions.
Though the original motivation was to give a proof of the triangle inequality as simple as possible, 
the link with submodular functions is interesting in itself (as also recently suggested in \cite{gillenwater-etal-2015}).



Let $X$ be a finite, non-empty ground set.
A set function $f:{\cal P}(X)\to \mathbb{R}$ is said to be {\em submodular} on $X$ if 
$f(A\cup B) + f(A\cap B) \le f(A) + f(B)$ for all $A,B\subseteq X$.
If all inequalities are equations then $f$ is called {\em modular} on $X$.
It is known that $f$ is submodular on $X$ if and only if the following condition holds (cf., e.g., \cite{schrijver-2003b}):
\begin{eqnarray}
f(A\cup \{x\}) - f(A) \ge f(B\cup\{x\}) - f(B) \qquad \textrm{for all $A\subseteq B\subseteq X$, $x\in \overline{B}$}\label{def:submodular2}
\end{eqnarray}
A set function $f$ is monotone if  $f(A)\le f(B)$ for all $A\subseteq B\subseteq X$; 
$f$ is nonnegative if $f(A)\ge 0$ for all $A\subseteq X$. 
Each nonnegative, monotone, modular function $f$ on $X$ can be written as $f(A)= \gamma + \sum_{i\in A} c_i$
where $\gamma, c_i\ge 0$ for all $i\in X$ (cf., e.g., \cite{schrijver-2003b}). Examples are set cardinality or
degree sum in graphs.
Standard examples of nonnegative, monotone, submodular set functions are matroid rank, network flow to a sink, 
entropy of sets of random variables, and  neighborhood size in bipartite graphs.

Let $f$ be a nonnegative, monotone, submodular set function on $X$.
For sets $A,B\subseteq X$, we define 
two candidates for {\em submodular Jaccard distances}, $J_{\delta,f}$ and $J^\Delta_{\delta,f}$, as follows:
\[
J_{\delta,f}(A,B)=_{\rm def} 1-\frac{f(A\cap B)}{f(A\cup B)},\qquad J^\Delta_{\delta,f} =_{\rm def} 
\frac{f(A \triangle B)-f(\emptyset)}{f(A\cup B)}, \]
where $J_{\delta,f}(A,B)=J^\Delta_{\delta,f}(A,B)=_{\rm def} 0$ if $f(A\cup B)=0$.
It is clear that $0\le J_{\delta,f}(A,B)\le J^\Delta_{\delta,f}(A,B)$.
If $f$ is modular then $J_{\delta,f}= J^\Delta_{\delta,f}$. In particular,
for $f(A)=|A|$ (i.e., the cardinality of the set $A\subseteq X$), we obtain the standard Jaccard distance $J_\delta=J_{\delta,f}=J^\Delta_{\delta,f}$. 

First, we give a simple proof of the triangle inequality for $J_{\delta,f}$. 
Interestingly, this is only possible for modular set functions (see the third remark after Theorem \ref{theorem:mengen:exkurs-mengen}).


\begin{lemma}\label{lemma:mengen:exkurs-mengen}
Let $f$ be a nonnegative, monotone, submodular set function on $X$. Then, for all sets $A,B,C\subseteq X$, it holds that
\[ f(A\cap C)\cdot f(B\cup C) + f(A\cup C)\cdot f(B\cap C) ~\le~ f(C) \cdot \bigl(f(A)+f(B)\bigr).\]
\end{lemma}

\begin{proof} We easily obtain
\begin{eqnarray*}
\lefteqn{f(A\cap C) \cdot f(B\cup C)}~~~~~~~~~\\
&\le&  f(A\cap C) \cdot \bigl(f(B)+f(C)-f(B\cap C)\bigr) ~~~~~~~~~~~~~\, \textrm{(submodularity of $f$)}\\
&\le&  f(C)\cdot \bigl(f(B)-f(B\cap C) + f(A\cap C)\bigr) ~~~~~~~~~~~~~~~\textrm{(monotonicity of $f$)}
\end{eqnarray*}
and, by swapping $A$ and $B$, $f(A\cup C) \cdot f(B\cap C)\le f(C)\cdot \bigl(f(A)-f(A\cap C) + f(B\cap C)\bigr).$
Overall,
\begin{eqnarray*}
\lefteqn{f(A\cap C) \cdot f(B\cup C) + f(A\cup C) \cdot f(B\cap C)}~~~~~~~ \\
	&\le& f(C)\cdot \bigl(f(B)-f(B\cap C) + f(A\cap C) + f(A)-f(A\cap C)+f(B\cap C)\bigr) \\
	&=& f(C)\cdot \bigl(f(B) + f(A) \bigr)
	\end{eqnarray*}
This shows the lemma.
\end{proof}



\begin{corollary}\label{corollary:mengen:exkurs-mengen}
Let $f$ be a nonnegative, monotone, submodular set function on $X$. Then, for all sets $S,T\subseteq X$, it holds that
\[f(S\cap T)\cdot f(S\cup T) \le f(S)\cdot f(T).\]
\end{corollary}

\begin{proof} Apply Lemma \ref{lemma:mengen:exkurs-mengen} to sets $A=_{\rm def} S$, $B=_{\rm def} S$ and $C=_{\rm def} T$.
\end{proof}


\begin{theorem}\label{theorem:mengen:exkurs-mengen}
Let $f$ be a nonnegative, monotone, modular set function on $X$. 
Then, for all sets $A,B,C\subseteq X$, it holds that
\[J_{\delta,f}(A,B) \le J_{\delta,f}(A,C)+ J_{\delta,f}(C,B).\]
\end{theorem}

\begin{proof}
Say that a set $A$ is a null set iff $f(A)=0$. 
Observe that if at least one of the sets is a null set then the inequality is satisfied.
So, it is enough to show the equivalent inequality
\begin{eqnarray}
\frac{f(A\cap C)}{f(A\cup C)} + \frac{f(B\cap C)}{f(B\cup C)} ~\le~ 1+ \frac{f(A\cap B)}{f(A\cup B)} = \frac{f(A)+f(B)}{f(A\cup B)}\label{ineq2}
\end{eqnarray}
for arbitrary non-null sets $A,B,C\subseteq I$. This is seen as follows:
\begin{eqnarray*}
\lefteqn{\frac{f(A\cap C)}{f(A\cup C)} + \frac{f(B\cap C)}{f(B\cup C)}} ~~~~~ \\[1ex]
&=& \frac{f(A\cap C)\cdot f(B\cup C) + f(A\cup C)\cdot f(B\cap C)}{f(A\cup C)\cdot f(B\cup C)} \\[1ex]
&\le& \frac{f(C)\cdot\bigl(f(A)+f(B)\bigr)}{f(A\cup C)\cdot f(B\cup C)}
~~~~~~~~~~~~~~~~~~~~~~~~~~~~~~~~~~~~~~~~~~~~~~~~~ (\textrm{by Lemma \ref{lemma:mengen:exkurs-mengen}})\\[1ex]
&\le& \frac{f(C)\cdot\bigl(f(A)+f(B)\bigr)}{f\bigl((A\cup C)\cap (B\cup C)\bigr)\cdot f(A\cup B\cup C)}
~~~~~~~~~~~~~~~~~~~~~~~~~\ (\textrm{by Corollary \ref{corollary:mengen:exkurs-mengen}})\\[1ex]
&\le& \frac{f(C)}{f\bigl((A\cap B) \cup C\bigr)} \cdot \frac{f(A)+f(B)}{f(A\cup B)}
~~~~~~~~~~~~~~~~~~~~~~~~~~~~~~ (\textrm{monotonicity of $f$})\\[1ex]
&\le& \frac{f(A)+f(B)}{f(A\cup B)}
~~~~~~~~~~~~~~~~~~~~~~~~~~~~~~~~~~~~~~~~~~~~~~~~~~~~~ (\textrm{monotonicity of $f$})
\end{eqnarray*}
This proves the theorem.
\end{proof}

\noindent
{\em Remarks}: We comment on the proof of the triangle inequality for $J_{\delta,f}$:
\begin{enumerate}
\item It follows from Theorem \ref{theorem:mengen:exkurs-mengen} that the triangle inequality is valid for 
	the standard Jaccard distance $J_\delta$, 
	the generalized Jaccard distance given for vectors $x,y\in\mathbb{R}^n$ by
	\[1-\frac{\sum_{i=1}^n \min\ \{x_i,y_i\}}{\sum_{i=1}^n \max\ \{x_i, y_i\}}\]
	(with the subcase that $x_i=\mu_A(z)$ and $y_i=\mu_B(z)$ denote multiplicities of (occurrences of) $z$ 
	in multisets $A$ and $B$; cf.~\cite{kosters-laros-2007}), and the Steinhaus distance 
	\cite{marczewski-steinhaus-1958,gardner-etal-2014} (i.e., any set measures, including probability measures).
	We mention that all these results can equally easily be proven by the arguments in \cite{gilbert-1972}; however, 
	for modular functions satisfying $f(\emptyset)>0$, these arguments fail.
\item Theorem \ref{theorem:mengen:exkurs-mengen} is true for nonnegative, monotone, modular functions defined over 
	distributive lattices; Lemma \ref{lemma:mengen:exkurs-mengen} and Corollary \ref{corollary:mengen:exkurs-mengen}
	also hold for nonnegative, monotone, submodular functions defined over distributive lattices. 
	Notice that $J^\Delta_{\delta,f}$ is not defined over all distributive lattices (see also the third remark after Theorem \ref{theorem:mengen:exkurs-mengen2}).
\item In general, Theorem \ref{theorem:mengen:exkurs-mengen} is not true for nonnegative, monotone, submodular functions: 
	Any set function $f$ such that $f(A)=f(B)=f(A\cup B)>f(A\cap B)\ge 0$ for non-empty, incomparable sets 
	$A,B$ refutes $J_{\delta,f}(A,B)\le J_{\delta,f}(A,A\cup B)+J_{\delta,f}(A\cup B,B)$.
	Concrete examples include linear cost functions with budget restrictions, i.e., $f(A)=\min\{B, \sum_{i\in A} c_i\}$,
	or the neighborhood size in a bipartite graph $G=(U\uplus V,E)$, i.e., $f(A)=|\Gamma(A)|$ where $A\subseteq U$ and 
	$\Gamma(A)=\bigcup_{u\in A}\{v\in V|\{u,v\}\in E\}$.
\end{enumerate}

Next we give a simple proof of the triangle inequality for $J^\Delta_{\delta,f}$.

\begin{theorem}\label{theorem:mengen:exkurs-mengen2}
Let $f$ be a nonnegative, monotone, submodular set function on $X$. 
Then, for all sets $A,B,C\subseteq X$, it holds that
\[J^\Delta_{\delta,f}(A,B) \le J^\Delta_{\delta,f}(A,C)+ J^\Delta_{\delta,f}(C,B).\]
\end{theorem}

\begin{proof} We split the set $C$ into two disjoint sets $C_0\subseteq A\cup B$ and $C_1\subseteq \overline{A\cup B}$, 
both possibly empty, such that $C=C_0\cup C_1$. We obtain
\begin{eqnarray*}
\lefteqn{\frac{f(A\triangle C)-f(\emptyset)}{f(A\cup C)}+\frac{f(B\triangle C)-f(\emptyset)}{f(B\cup C)}} ~~~~~\\[1ex]
&\ge& \frac{f(A\triangle C)+f(B\triangle C)-2f(\emptyset)}{f(A\cup B\cup C_1)} ~~~~~~~~~~~~~~~~~~~~~~~~~~~~~~~~ (\textrm{monotonicity of $f$}) \\[1ex]
&\ge& \frac{f(A\triangle C \cup  B\triangle C)-f(\emptyset)}{f(A\cup B\cup C_1)}~~~~~~~~~~~~~~~~~~(\textrm{submodularity, monotonicity of $f$}) \\[1ex]
&\ge& \frac{f(A\triangle B \cup C_1)-f(\emptyset)}{f(A\cup B\cup C_1)} ~~~~~~~~~~~~~~~~~~~~~~~~~~~~~~~~~~~~~~~~~~ (\textrm{monotonicity of $f$}) \\[1ex]
&\ge& \frac{f(A\triangle B)}{f(A\cup B)}-\frac{f(\emptyset)}{f(A\cup B\cup C_1)} ~~~~~~~~~~~~~~~~~\ (\textrm{submodularity of $f$, Cond.~(\ref{def:submodular2})})\\[1ex]
 &\ge& \frac{f(A\triangle B)}{f(A\cup B)}-\frac{f(\emptyset)}{f(A\cup B)} ~~~~~~~~~~~~~~~~~~~~~~~~~~~~~~~~~~~~~~~~~ (\textrm{monotonicity of $f$})
\end{eqnarray*}
This shows the theorem.
\end{proof}

\noindent
{\em Remarks}: We comment on the proof of the triangle inequality for $J^\Delta_{\delta,f}$:
\begin{enumerate}
\item It follows once more from Theorem \ref{theorem:mengen:exkurs-mengen2} that the standard Jaccard distance, 
	the generalized Jaccard distance, and the Steinhaus distance satisfy the triangle inequality. 
	Moreover, $J^\Delta_{\delta,f}$ is also a (pseudo)metric for, e.g., linear cost functions with budget restrictions
	and the neighborhood size in bipartite graphs.
\item Theorem \ref{theorem:mengen:exkurs-mengen2} suggests that $J^\Delta_{\delta,f}$ is the right
	definition of a submodular Jaccard distance. As a consequence, one might say that the submodular
	Jaccard (similarity) index should be defined as the inverse submodular Jaccard distance, i.e.,
	\[J^\Delta_f(A,B)=_{\rm def} 1- J^\Delta_{\delta,f} = 1-\frac{f(A\triangle B)-f(\emptyset)}{f(A\cup B)}\]
	Again, if $f(A)=|A|$ then we obtain the standard Jaccard index $J=J^\Delta_f=1-J_{\delta,f}$.
\item 
	Though $J^\Delta_{\delta,f}$ might generally not be defined over a given distributive lattice, it can be seen 
	that for each nonnegative, monotone, submodular function $f:\XF\to\mathbb{R}$ defined on a family
	$\XF\subseteq {\cal P}(X)$ closed under union and intersection, there is a (not necessarily unique) 
	nonnegative, monotone, submodular extension $\overline{f}:{\cal P}(X)\to \mathbb{R}$ on $X$ such that 
	$\overline{f}(A)=f(A)$ for all $A\in \XF$ (e.g., \cite{topkis-1978}), so that $J^\Delta_{\delta,\overline{f}}$ 
	can be used instead.
\end{enumerate}

\bigskip

\noindent
{\small
{\bf Acknowledgments:} I am grateful to Ulrik Brandes (Konstanz) and Julian M\"uller (Konstanz) for helpful discussions.
}

\begin{small}

\end{small}


\begin{thebibliography}{10}

\bibitem{charikar-2002}
M.~S.~Charikar.
\newblock Similarity Estimation Techniques from Rounding Algorithms.
\newblock In: {\em Proceedings of the 34th Annual ACM Symposium on Theory of Computing (STOC'2002)}, pp.~380--388. ACM Press, New York, NY, 2002.

\bibitem{deza-deza-2009}
M.~M.~Deza, E.~Deza.
\newblock {\em Encyclopedia of Distances.}
\newblock Springer, Berlin, 2009.


\bibitem{fujita-2013}
O.~Fujita.
\newblock Metrics based on average distance between sets.
\newblock {\em Japan Journal of Industrial and Applied Mathematics}, 30(1):1--19, 2013. 

\bibitem{gardner-etal-2014}
A.~Gardner, J.~Kanno, C.~A.~Duncan, R.~Selmic.
\newblock Measuring Distance Between Unordered Sets of Different Sizes.
\newblock In: {\em Proceedings of the 2014 IEEE Conference on Computer Vision and Pattern Recognition (CVPR'2014)}, pp.~137--143. IEEE, New Jersey, NJ, 2014.


\bibitem{gilbert-1972}
G.~Gilbert.
\newblock Distance between sets.
\newblock Letters to {\em Nature}, 239(5368):174, 1972. 

\bibitem{gillenwater-etal-2015}
J.~Gillenwater, R.~Iyer, B.~Lusch, R.~Kidambi, J.~A.~Bilmes.
\newblock Submodular Hamming Metrics.
\newblock In: {\em Advances in Neural Information Processing Systems 28}, 3141-3149. NIPS Proceedings, December 2015.  

\bibitem{gower-2008}
J.~C.~Gower.
\newblock Similarity, Dissimilarity and Distance, Measures of.
\newblock In: S.~Kotz, C.~B.~Read, N.~Balakrishnan, B.~Vidakovic (eds.), {\em Encyclopedia of Statistical Sciences}, vol. 12., pp.~7730--7738. 2nd edition, John Wiley, New York, NY, 2008.


\bibitem{jaccard-1901}
P.~Jaccard.
\newblock \'{E}tude comparative de la distribution florale dans une portion des Alpes et du Jura.
\newblock {\em Bulletin de la Soci\'{e}t\'{e} Vaudoise des Sciences Naturelles}, 37(142):547--579, 1901. 

\bibitem{kosters-laros-2007}
W.~A.~Kosters, J.~F.~J.~Laros.
\newblock Metrics for Mining Multisets.
\newblock In: M.~Bramer, F.~Coenen, M.~Petridis (eds.), {\em Research and Development in Intelligent Systems XXIV, Proceedings
               of the Twenty-seventh {SGAI} International Conference on
               Innovative Techniques and Applications of Artificial Intelligence (AI'2007)}, pp.~293--303.
               Springer, Berlin, 2007.

\bibitem{levandowsky-winter-1971}
M.~Levandowsky, D.~Winter.
\newblock Distance between sets.
\newblock Letters to {\em Nature}, 234(5323):34--35, 1971. 

\bibitem{lipkus-1999}
A.~H.~Lipkus.
\newblock A proof of the triangle inequality for the Tanimoto distance.
\newblock {\em Journal of Mathematical Chemistry}, 26:263--265, 1999.

\bibitem{marczewski-steinhaus-1958}
E.~Marczewski, H.~Steinhaus.
\newblock On a certain distance of sets and the corresponding distance of functions.
\newblock {\em Colloquium Mathematicum}, 6:319--327, 1958.

\bibitem{simovici-djeraba-2008}
D.~A.~Simovici, C.~Djeraba.
\newblock {\em Mathematical Tools for Data Mining}.
\newblock Springer, London, 2008.

\bibitem{schrijver-2003b}
A.~Schrijver.
\newblock {\em Combinatorial Optimization}, vol. B.
\newblock Springer, Berlin, 2003.

\bibitem{tanimoto-1958}
T.~T.~Tanimoto.
\newblock An elementary mathematical theory of classification and prediction.
\newblock IBM Report, November 1958.

\bibitem{topkis-1978}
D.~M.~Topkis.
\newblock Minimizing a submodular function on a lattice.
\newblock {\em Operations Research}, 26(2):305--321, 1978.


\end{thebibliography}
\end{document}